\documentclass[11pt, oneside]{article}   	
\usepackage{fullpage}
\usepackage[pass]{geometry} 
\usepackage{graphicx}				
\usepackage{amsmath,amssymb,amsthm,enumitem}

\usepackage[bookmarks=true,hypertexnames=false,pagebackref]{hyperref}
\hypersetup{colorlinks=true, citecolor=blue, linkcolor=red, urlcolor=blue}

\usepackage{cleveref}
\usepackage{tikz}
\usetikzlibrary{arrows,arrows.meta,backgrounds,calc,fit,decorations.pathreplacing,decorations.markings,shapes.geometric}
\tikzset{>=latex}

\newtheorem{theorem}{Theorem}
\newtheorem{lemma}[theorem]{Lemma}

\newtheorem{proposition}[theorem]{Proposition}
\newtheorem{example}[theorem]{Example}
\newtheorem{problem}[theorem]{Problem}

\def\NP{\ensuremath{\mathbf{NP}}}
\def\RP{\ensuremath{\mathbf{RP}}}
\def\Ptime{\ensuremath{\mathbf{P}}}
\def\BPP{\ensuremath{\mathbf{BPP}}}
\def\numP{\textup{\#}\mathbf{P}}

\def\*#1{\mathbf{#1}}
\def\+#1{\mathcal{#1}}
\def\-#1{\mathrm{#1}}
\def\=#1{\mathbb{#1}}

\def\bfz{\mathbf{0}}
\def\bfo{\mathbf{1}}
\def\Nset{\mathbb{N}}
\def\Zset{\mathbb{Z}}
\def\Rset{\mathbb{R}}
\def\nIS{\textsc{\#IS}}
\def\nBIS{\textsc{\#BIS}}

\def\nPolytopeVertices{\textsc{\#Vertices of a 0/1 Polytope}}

\def\npnSAT{\textsc{\#1p1nSAT}}
\def\Zptm{Z_\textup{P2M}}
\def\Zptmk#1{Z_{\textup{P2M},#1}}
\def\Zoc{Z_\textup{OC}}

\newcommand{\transpose}[1]{#1^{\texttt T}}
\newcommand{\abs}[1]{\left\vert#1\right\vert}
\let\implies=\Rightarrow
\newcommand{\cupdot}{\mathbin{\mathaccent\cdot\cup}}
\let\epsilon=\varepsilon

\title{Counting vertices of integral polytopes defined by facets}
\author{Heng Guo\thanks{HG has received funding from the European Research Council (ERC) under the European Union's Horizon 2020 research and innovation programme (grant agreement No.~947778).}\\School of Informatics\\University of Edinburgh\\\url{hguo@inf.ed.ac.uk} \and 
Mark Jerrum\thanks{MJ was supported by grant EP/S016694/1 `Sampling in hereditary classes' from the Engineering and Physical Sciences Research Council (EPSRC) of the UK.}\\School of Mathematical Sciences\\Queen Mary, University of London\\\url{m.jerrum@qmul.ac.uk}}
\date{}							

\begin{document}
\maketitle

\begin{abstract}
We present a number of complexity results concerning the problem of counting vertices of an integral polytope defined by a system of linear inequalities.  The focus is on polytopes with small integer vertices, particularly 0/1 polytopes and half-integral polytopes.  
\end{abstract}

\subsubsection*{Keywords}0/1 polytopes, approximation algorithms, computational complexity of counting, totally unimodular matrices. 

\subsubsection*{MSC} 68Q17 (Computational difficulty of problems), 52B05 (Combinatorial properties of polytopes and polyhedra), 68W25 (Approximation algorithms).

\subsubsection*{Data availability statement} Data sharing not applicable to this article as no datasets were generated or analysed during the current study.

\section{Introduction}

We are interested in counting (or uniformly sampling) vertices of a polytope defined by linear inequalities $Ax\leq b$.  In particular we concentrate on the sort of polytopes that arise in the study of computational optimisation problems: integral polytopes (convex hulls of finite subsets of $\Nset^n$) and especially $0/1$ polytopes (convex hulls of subsets of $\{0,1\}^n$).  We assume that instances are presented as systems of linear inequalities. 

\goodbreak

\begin{description}[itemsep=0pt]
\item [Problem:] \nPolytopeVertices.
\item [Instance:]  An $m\times n$ integer matrix $A\in\Zset^{m\times n}$ and a vector $b\in\Zset^m$.
\item [Promise:] The inequalities $Ax\leq b$ define a 0/1 polytope $P$.
\item [Output:]  The number of vertices of $P$.
\end{description}

It is not clear whether it is possible to decide, in polynomial time, if a given system of linear inequalities defines a 0/1 polytope.
Papadimitriou and Yannakakis~\cite{PapadimitriouYannakakis} have proved that deciding integrality of general polytopes is co-NP-complete, but the restiction to the 0/1 case may be easier.  
By combining several technically involved results, Ding et al.~\cite{DingFengZang} show that integrality of the polytope $\{x\in\Rset^n:Ax\leq b, x\geq\bfz\}$ can be decided in polynomial time, in the case where $A$ is a 0/1 matrix and $b=\bfo$.  
Given the uncertainty surrounding this question, it seems reasonable to add the promise that $P$ is a 0/1 polytope.  Note that checking that a polytope~$P$ is 0/1 essentially reduces to checking that $P$ is integral, as containment of~$P$ in the cube $[0,1]^n$ can be tested efficiently by linear programming.

In cases of interest, the matrix $A$ and vector $b$ will not be arbitrary, but will have been introduced with a certain goal in mind.  The insights that led to their construction can almost certainly be used to provide an elementary proof of integrality of~$P$.  Furthermore, except in Section~\ref{sec:TDI}, the matrix $A$ will be `totally unimodular', and hence will necessarily define an integral polytope.  Total unimodularity is decidable in polynomial time using Seymour's decomposition theorem for regular matroids~\cite{SeymourDecomp}:  see Truemper~\cite{Truemper}.  

Our first observation is that \nPolytopeVertices{} is hard to solve \emph{exactly}. This follows easily by encoding perfect matchings in bipartite graphs as vertices of a $0/1$ polytope. A proof of this easy result can be found at the end of this section.

\begin{proposition}\label{prop:numPcomplete}
\nPolytopeVertices\ is $\numP$-complete.
\end{proposition}

In light of this strong negative result, we naturally turn our attention to approximate counting.  Before presenting our results, we briefly review the main definitions and concepts used in the study of approximation algorithms for counting problems.   The reader is directed to Dyer, Goldberg, Greenhill and Jerrum~\cite{Relative} for precise definitions and a survey of the wider context. 

The standard notion of efficient approximation algorithm in the context of counting problems is the {\it Fully Polynomial Randomised Approximation Scheme}, or FPRAS{}.  This is a randomised algorithm that is required to produce a solution within relative error $1\pm\varepsilon$,\footnote{For many problems, including those considered here, a polynomial approximation ratio can be efficiently amplified to a $(1\pm\varepsilon)$ approximation ratio.} with probability at least $\frac34$,\footnote{The success probability can be raised from $\frac34$ to arbitrarily close to~1 by running the algorithm a number of times and taking the median of the outputs.}  in time polynomial in the instance size and $\varepsilon^{-1}$.
A deterministic algorithm that achieves the same end without error is called a {\it Fully Polynomial Time Approximation Scheme\/} or FPTAS{}.  The computational complexity of approximate counting problems can be compared through {\it Approximation-Preserving\/} (or AP-) {\it reductions}.  These are polynomial-time Turing reductions that preserve (closely enough) the error tolerance.\footnote{Although AP-reductions are generally allowed to be randomised, the ones employed here are all deterministic.}  The set of problems that have an FPRAS is closed under AP-reducibility.  

Stockmeyer~\cite{Stockmeyer} was the first to produce evidence that approximate counting is of lower computational complexity than exact counting.  This key insight was refined by Valiant and Vazirani~\cite[Cor.~3.6]{NPeasy}, who showed that every function in $\numP$ can be approximated (in the FPRAS sense) by a polynomial-time Turing machine with an oracle for an \NP-complete problem.  Therefore the strongest negative result we can have for approximate counting is one of \NP-hardness.  An example of such a maximally hard problem is $\nIS$, which asks for the number of independent sets of all sizes in a general graph.  It follows that the existence of an FPRAS for $\nIS$ would imply $\RP=\NP$, and the existence of an FPTAS would imply $\Ptime=\NP$.  An \NP-hard problem concerning polytopes is given in the next section.  

Many counting problems are known to have an FPRAS and many others to be \NP-hard to approximate.  An interesting empirical observation is that many of the rest are interreducible via approximation-preserving reductions.  One member of this equivalence class is $\nBIS$, which asks for the number of independent sets of all sizes in a bipartite graph.  
\begin{description}[itemsep=0pt]
\item [Problem:] \nBIS.
\item [Instance:]  A bipartite graph $B$.
\item [Output:]  The number of independent sets (of all sizes) in $B$.
\end{description}
After two decades of fairly intensive study, no FPRAS for $\nBIS$ has been discovered, but neither has $\nBIS$ been shown to be \NP-hard to approximate.  So showing that a counting problem $\Pi$ is approximation-preserving interreducible with $\nBIS$ can be interpreted as evidence that $\Pi$ does not admit an FPRAS, though the evidence falls short of a demonstration of \NP-hardness.  An example involving 0/1 polytopes is presented in Section~\ref{sec:transposenetwork}.

Khachiyan~\cite{Khachiyan} showed that approximately counting vertices of general polytopes presented as linear inequalities is \NP-hard, a result rediscovered by Najt \cite{Najt}.  As we just saw, this is the strongest possible demonstration of intractability.  However, the vertices of the polytopes employed in the proof have rational coordinates, which, if rescaled to integers, would be exponential in the dimension of the ambient space.  Our aim in this work is to find hard examples in small integers.

Ultimately, we would like to characterise the complexity of computing approximate solutions to \nPolytopeVertices.  We have not been able to establish the strongest intractability result, which would be a demonstration of \NP-hardness.  
The typical approach of showing \NP-hardness for an approximate counting problem is to reduce from a combinatorial optimisation problem.
However, the combinatorial optimisation problem related to a $0/1$ polytope is often tractable, which makes showing \NP-hardness difficult.

Instead, we are able to obtain an \NP-hardness result by relaxing the allowed problem instances from 0/1 polytopes to `half-integral' polytopes whose vertices are elements of $\{0,1,2\}^n$.  (More conventionally, these polytopes are scaled by a factor 2, so as to have vertices in $\{0,\frac12,1\}^n$.)  This negative result relates to a natural family of `perfect 2-matching polytopes' associated with graphs.  See Proposition~\ref{prop:P2M-vertices}.

What can be said about the complexity of \nPolytopeVertices{} itself?  The most prominent class of matrices defining integral polytopes are the \emph{totally unimodular matrices}.  These are matrices~$A$ that have the property that every square submatrix of $A$ has determinant $-1$, 0 or 1.  (In particular, the elements of $A$ take values in $\{-1,0,1\}$.)  Many 0/1 polytopes arising in combinatorial optimisation arise from totally unimodular matrices.  Network matrices and transposes of network matrices are natural subclasses of totally unimodular matrices that will be defined in Section~\ref{sec:transposenetwork}.  In some sense, these matrix classes are universal in that every totally unimodular matrix can be built from network matrices, transposes of network matrices, and a certain $5\times5$ matrix.  

In Section~\ref{sec:transposenetwork}, we show (Theorem~\ref{thm:BISequiv}) that \nPolytopeVertices, when restricted to transposes of network matrices, is interreducible with $\nBIS$ with respect to approximation-preserving reductions.  This locates this special case of the problem as accurately as possible, given the current state of knowledge of the complexity landscape.   When restricted to network matrices, \nPolytopeVertices\ appears to be become easier, and we identify a subclass of polytopes (Proposition~\ref{prop:network}) for which the vertex counting problem is solvable in the FPRAS sense.  We leave it as an open question whether there is an FPRAS for \nPolytopeVertices\ when restricted to network matrices more generally.  

In the final section we go beyond totally unimodular matrices.  Here, there are fewer examples in the literature, but we do note that at least one naturally occurring class of polytopes, arising from `stable matchings', has been shown to give rise to a vertex counting problem that is interreducible with $\nBIS$.  This raises the intriguing possibility that \nPolytopeVertices\ itself is actually equivalent in complexity to $\nBIS$.

The complexity of approximate counting and (almost) uniform sampling are usually closely related, and this is indeed the case here.  For simplicity we concentrate throughout on approximate counting, but the corresponding uniform sampling problems have essentially the same complexity.  This follows by general considerations~\cite{JVV} from the fact that \nPolytopeVertices\ and all the restrictions of it considered here are self-reducible.  We expand on this remark at the end of Section~\ref{sec:transposenetwork}.   

Mihail and Vazirani \cite{MV89} conjectured that the simple random walk on the graph of a $0/1$ polytope is rapidly mixing.  This conjecture remains open, though it has been shown to hold in important special cases, for example, the polytope defined by the bases of a matroid~\cite{ALOV}. Note that this conjecture, if true, does not directly imply an FPRAS for \nPolytopeVertices{}, because the degree of the vertices can be exponentially large in the dimension of the ambient space. The stationary distribution of the random walk over graphs of $0/1$ polytopes can be very different from the uniform distribution.

Although we restrict attention to complexity-theoretic results in this work, it should be noted that several authors have studied heuristic approaches, including Avis and Devroye~\cite{AvisDevroye} and Salomone, Vaisman and Kroese~\cite{SalomoneEtAl}.  We round off the section with the deferred proof. 

\begin{proof}[Proof of Proposition \ref{prop:numPcomplete}]
Membership in $\numP$ is clear.  

To demonstrate hardness, we simply show that \nPolytopeVertices\ includes counting perfect matchings in a bipartite graph as a special case.  Suppose $G$ is a bipartite graph, and denote the vertex and edge sets of~$G$ by $V(G)$ and $E(G)$. Introduce variables $\{x_{uv}:uv\in E(G)\}$ in 1-1 correspondence with the $m$ edges of~$G$.  Consider the $m$-dimensional perfect matching polytope $P_\mathrm{PM}(G)$ of~$G$, whose vertices are in bijection with the perfect matchings of~$G$.  Specifically, for each perfect matching $M\subseteq E(G)$ of $G$ there corresponds a vertex of $P_\mathrm{PM}(G)$ given by 
$$
x_{uv}=\begin{cases}
1,&\text{if $uv\in M$};\\
0,&\text{otherwise},
\end{cases}
$$
for all $uv\in E(G)$.  In the case that $G$ is bipartite, the polytope $P_\mathrm{PM}(G)$ is defined by the following inequalities~\cite[Thm 18.1]{SchrijverA}:
\begin{align*}
0\leq x_{uv}\leq1,&\quad \text{for all $uv\in E(G)$, and}\\
\sum_{v\in V(G):uv\in E(G)} x_{uv}=1,&\quad \text{for all $u\in V(G)$}.
\end{align*}
Thus, counting perfect matchings in a bipartite graph can be reduced to counting vertices of an easily computable and easily described polytope.  The result follows from Valiant's classical result that counting perfect matchings in a bipartite graph is $\numP$-complete~\cite{ValPerm}. 
\end{proof}

\section{The perfect 2-matching polytope}
In this section we see that by going a little beyond 0/1 polytopes we can find counting problems that are \NP-hard to approximate.

Given a graph $G$, the \emph{perfect 2-matching polytope} (P2M polytope) is defined by the system of linear inequalities
\begin{align*}
  0\leq x_{uv}\leq 2, & \quad\text{for $uv\in E(G)$};\\
  \sum_{v:uv\in E(G)} x_{uv} = 2, & \quad\text{for $u\in V(G)$}.
\end{align*}
Let $\Zptm(G)$ be the number of vertices of the P2M polytope associated with $G$.
We will show that approximating $\Zptm(G)$ is \NP-hard.



Recall that an edge cover of a graph is a set $C$ of edges such that any vertex is incident to some edge $e\in C$.
Balinski \cite{Balinski} observed the following characterisation of the vertices of the P2M polytope.
(See also Schrijver~\cite[Cor.~30.2b]{SchrijverA} together with the observation at the end of \cite[\S30.3]{SchrijverA}.)

\begin{proposition}  \label{prop:P2M-vertices}
  The vertices of the P2M polytope correspond to edge covers consisting of a matching $M$, 
  with $x_e=2$ for $e\in M$, and vertex-disjoint odd cycles that are also vertex-disjoint from $M$, 
  with $x_e=1$ for each edge $e$ in any of the odd cycles.
\end{proposition}

\begin{theorem} \label{thm:P2M-NP-hard} The problem of estimating $\Zptm(G)$ with constant relative error is hard for $\NP$ with respect to polynomial-time reducibility.  Thus, there is no FPRAS for $\Zptm(G)$ unless $\RP=\NP$ and no FPTAS unless $\Ptime=\NP$.
\end{theorem}
\begin{proof}
  We call edge covers corresponding to the vertices of the P2M polytope, as in \Cref{prop:P2M-vertices}, \emph{P2M covers},
  and call edge covers of $G$ consisting of vertex-disjoint odd cycles without a matching \emph{odd cycle covers}.  We denote the number of odd cycle covers of~$G$ by $\Zoc(G)$.  By definition, the number of P2M covers is $\Zptm(G)$.

  First we reduce deciding the existence of a Hamiltonian path between two given vertices in a bipartite graph to deciding the existence of an odd cycle cover in a general graph.  The former problem is known to be \NP-complete~\cite{bipartiteHC}.  (The given reference treats the Hamilton \emph{cycle} problem, but essentially the same reduction deals with \emph{paths}.)
  Let $G$ be a bipartite graph with two distinguished vertices $s$ and~$t$ on opposite sides of the bipartition.
  We introduce a new vertex~$w$, and add two edges $sw$ and $wt$.
  Call the resulting graph~$G'$.  
  Any odd cycle in $G'$ must include the new vertex~$w$, as $G'-w$ is bipartite.
  It follows that any odd cycle cover of~$G'$ must consist of a Hamiltonian path in $G$ from $s$ to~$t$, together with the two edges incident at~$w$.  
  Conversely, any Hamiltonian path in $G$ from $s$ to~$t$ can be extended to an odd cycle cover of~$G'$ by adding edges $sw$ and $wt$.
  Thus, deciding the existence of an odd cycle cover of a graph is \NP-complete.

  Next we reduce deciding the existence of an odd cycle cover to approximating the number of P2M covers.
  For the reduction, we use the gadget in \Cref{fig:hex-gadget} to reduce the contribution from the matching (isolated edges).
  Note that the parameter~$\ell$ will be tuned later.

  \begin{figure}[htbp]
    \centering
    \begin{tikzpicture}[scale=0.8, inner sep=1pt, transform shape]
      \draw (0,0) node [draw,fill,shape=circle,color=black, label=180:{\Large $u$}] (u) {}; 
      \draw (0.5,0.866) node [draw,fill,shape=circle,color=black] (u1) {};
      \draw (1.5,0.866) node [draw,fill,shape=circle,color=black] (u2) {};
      \draw (0.5,-0.866) node [draw,fill,shape=circle,color=black] (u1') {};
      \draw (1.5,-0.866) node [draw,fill,shape=circle,color=black] (u2') {};
      \draw (2,0) node [draw,fill,shape=circle,color=black] (u3) {};

      \draw (u) edge [semithick] (u1)
      (u1) edge [semithick] (u2)
      (u2) edge [semithick] (u3)
      (u) edge [semithick] (u1')
      (u1') edge [semithick] (u2')
      (u2') edge [semithick] (u3);

      \draw (2.5,0.866) node [draw,fill,shape=circle,color=black] (u4) {};
      \draw (3.5,0.866) node [draw,fill,shape=circle,color=black] (u5) {};
      \draw (2.5,-0.866) node [draw,fill,shape=circle,color=black] (u4') {};
      \draw (3.5,-0.866) node [draw,fill,shape=circle,color=black] (u5') {};
      \draw (4,0) node [draw,fill,shape=circle,color=black] (u6) {};      

      \draw (u3) edge [semithick] (u4)
      (u4) edge [semithick] (u5)
      (u5) edge [semithick] (u6)
      (u3) edge [semithick] (u4')
      (u4') edge [semithick] (u5')
      (u5') edge [semithick] (u6);

      \draw (5,0) node {\LARGE \dots \dots};

      \draw (6,0) node [draw,fill,shape=circle,color=black] (u6') {};      
      \draw (6.5,0.866) node [draw,fill,shape=circle,color=black] (u7) {};
      \draw (7.5,0.866) node [draw,fill,shape=circle,color=black] (u8) {};
      \draw (6.5,-0.866) node [draw,fill,shape=circle,color=black] (u7') {};
      \draw (7.5,-0.866) node [draw,fill,shape=circle,color=black] (u8') {};
      \draw (8,0) node [draw,fill,shape=circle,color=black] (u9) {};      

      \draw (u6') edge [semithick] (u7)
      (u7) edge [semithick] (u8)
      (u8) edge [semithick] (u9)
      (u6') edge [semithick] (u7')
      (u7') edge [semithick] (u8')
      (u8') edge [semithick] (u9);

      \draw (8.5,0.866) node [draw,fill,shape=circle,color=black] (u10) {};
      \draw (9.5,0.866) node [draw,fill,shape=circle,color=black] (u11) {};
      \draw (8.5,-0.866) node [draw,fill,shape=circle,color=black] (u10') {};
      \draw (9.5,-0.866) node [draw,fill,shape=circle,color=black] (u11') {};
      \draw (10,0) node [draw,fill,shape=circle,color=black] (u12) {};      
      \draw (11,0) node [draw,fill,shape=circle,color=black, label=0:{\Large $v$}] (v) {};

      \draw (u9) edge [semithick] (u10)
      (u10) edge [semithick] (u11)
      (u11) edge [semithick] (u12)
      (u9) edge [semithick] (u10')
      (u10') edge [semithick] (u11')
      (u11') edge [semithick] (u12)
      (u12) edge [semithick] (v);   
    \end{tikzpicture}  
    \caption{The hexagon gadget, with $2\ell$ hexagons in the middle.}
    \label{fig:hex-gadget}
  \end{figure}
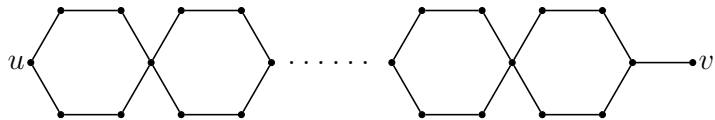  

  Let $G$ be an instance for odd cycle covers.  
  Replace each edge $uv\in E(G)$ with a copy of the gadget, and consider any P2M cover in the resulting graph $G'$.
  The possible configurations induced by this cover on the gadget with end points $u$ and $v$ can be partitioned into three types (see \Cref{fig:types-PMU}):
  \begin{itemize}
    \item Type P[ath].  There is a path from $u$ to $v$ together with some isolated edges.
      Note that there are two choices for the path in each hexagon.
    \item Type M[atched].  There is no path from $u$ to $v$, but $u$ and $v$ are both covered (by isolated edges).
      Note that there are two choices in every other hexagon, starting from the first hexagon (counting from $u$).
    \item Type U[nmatched].  There is no path from $u$ to $v$, and $u$ and $v$ are both uncovered.
      Note that there are two choices in every other hexagon, starting from the second hexagon (counting from $u$).
  \end{itemize}
  Note that this list is exhaustive, since it is impossible to match exactly one of $u,v$. 

  \begin{figure}[htbp]
    \centering
    \begin{tikzpicture}[scale=0.8, inner sep=1pt, transform shape]
      \draw (0,0) node [draw,fill,shape=circle,color=black, label=180:{\Large $u$}] (u) {}; 
      \draw (0.5,0.866) node [draw,fill,shape=circle,color=black] (u1) {};
      \draw (1.5,0.866) node [draw,fill,shape=circle,color=black] (u2) {};
      \draw (0.5,-0.866) node [draw,fill,shape=circle,color=black] (u1') {};
      \draw (1.5,-0.866) node [draw,fill,shape=circle,color=black] (u2') {};
      \draw (2,0) node [draw,fill,shape=circle,color=black] (u3) {};

      \draw (u) edge [thick, red] (u1)
      (u1) edge [thick, red] (u2)
      (u2) edge [thick, red] (u3)
      (u) edge [semithick] (u1')
      (u1') edge [thick, red] (u2')
      (u2') edge [semithick] (u3);

      \draw (2.5,0.866) node [draw,fill,shape=circle,color=black] (u4) {};
      \draw (3.5,0.866) node [draw,fill,shape=circle,color=black] (u5) {};
      \draw (2.5,-0.866) node [draw,fill,shape=circle,color=black] (u4') {};
      \draw (3.5,-0.866) node [draw,fill,shape=circle,color=black] (u5') {};
      \draw (4,0) node [draw,fill,shape=circle,color=black] (u6) {};      

      \draw (u3) edge [thick, red] (u4)
      (u4) edge [thick, red] (u5)
      (u5) edge [thick, red] (u6)
      (u3) edge [semithick] (u4')
      (u4') edge [thick, red] (u5')
      (u5') edge [semithick] (u6);

      \draw (5,0) node {\LARGE \dots \dots};

      \draw (6,0) node [draw,fill,shape=circle,color=black] (u6') {};      
      \draw (6.5,0.866) node [draw,fill,shape=circle,color=black] (u7) {};
      \draw (7.5,0.866) node [draw,fill,shape=circle,color=black] (u8) {};
      \draw (6.5,-0.866) node [draw,fill,shape=circle,color=black] (u7') {};
      \draw (7.5,-0.866) node [draw,fill,shape=circle,color=black] (u8') {};
      \draw (8,0) node [draw,fill,shape=circle,color=black] (u9) {};      

      \draw (u6') edge [semithick] (u7)
      (u7) edge [thick, red] (u8)
      (u8) edge [semithick] (u9)
      (u6') edge [thick, red] (u7')
      (u7') edge [thick, red] (u8')
      (u8') edge [thick, red] (u9);

      \draw (8.5,0.866) node [draw,fill,shape=circle,color=black] (u10) {};
      \draw (9.5,0.866) node [draw,fill,shape=circle,color=black] (u11) {};
      \draw (8.5,-0.866) node [draw,fill,shape=circle,color=black] (u10') {};
      \draw (9.5,-0.866) node [draw,fill,shape=circle,color=black] (u11') {};
      \draw (10,0) node [draw,fill,shape=circle,color=black] (u12) {};      
      \draw (11,0) node [draw,fill,shape=circle,color=black, label=0:{\Large $v$}] (v) {};

      \draw (u9) edge [thick, red] (u10)
      (u10) edge [thick, red] (u11)
      (u11) edge [thick, red] (u12)
      (u9) edge [semithick] (u10')
      (u10') edge [thick, red] (u11')
      (u11') edge [semithick] (u12)
      (u12) edge [thick, red] (v);

      \draw (5,-1.5) node {Type P};

      \begin{scope}[shift = {(0,-4)}]
        \draw (0,0) node [draw,fill,shape=circle,color=black, label=180:{\Large $u$}] (u) {}; 
        \draw (0.5,0.866) node [draw,fill,shape=circle,color=black] (u1) {};
        \draw (1.5,0.866) node [draw,fill,shape=circle,color=black] (u2) {};
        \draw (0.5,-0.866) node [draw,fill,shape=circle,color=black] (u1') {};
        \draw (1.5,-0.866) node [draw,fill,shape=circle,color=black] (u2') {};
        \draw (2,0) node [draw,fill,shape=circle,color=black] (u3) {};
  
        \draw (u) edge [thick, red] (u1)
        (u1) edge [semithick] (u2)
        (u2) edge [thick, red] (u3)
        (u) edge [semithick] (u1')
        (u1') edge [thick, red] (u2')
        (u2') edge [semithick] (u3);
  
        \draw (2.5,0.866) node [draw,fill,shape=circle,color=black] (u4) {};
        \draw (3.5,0.866) node [draw,fill,shape=circle,color=black] (u5) {};
        \draw (2.5,-0.866) node [draw,fill,shape=circle,color=black] (u4') {};
        \draw (3.5,-0.866) node [draw,fill,shape=circle,color=black] (u5') {};
        \draw (4,0) node [draw,fill,shape=circle,color=black] (u6) {};      
  
        \draw (u3) edge [semithick] (u4)
        (u4) edge [thick, red] (u5)
        (u5) edge [semithick] (u6)
        (u3) edge [semithick] (u4')
        (u4') edge [thick, red] (u5')
        (u5') edge [semithick] (u6);
  
        \draw (5,0) node {\LARGE \dots \dots};
  
        \draw (6,0) node [draw,fill,shape=circle,color=black] (u6') {};      
        \draw (6.5,0.866) node [draw,fill,shape=circle,color=black] (u7) {};
        \draw (7.5,0.866) node [draw,fill,shape=circle,color=black] (u8) {};
        \draw (6.5,-0.866) node [draw,fill,shape=circle,color=black] (u7') {};
        \draw (7.5,-0.866) node [draw,fill,shape=circle,color=black] (u8') {};
        \draw (8,0) node [draw,fill,shape=circle,color=black] (u9) {};      
  
        \draw (u6') edge [semithick] (u7)
        (u7) edge [thick, red] (u8)
        (u8) edge [semithick] (u9)
        (u6') edge [thick, red] (u7')
        (u7') edge [semithick] (u8')
        (u8') edge [thick, red] (u9);
  
        \draw (8.5,0.866) node [draw,fill,shape=circle,color=black] (u10) {};
        \draw (9.5,0.866) node [draw,fill,shape=circle,color=black] (u11) {};
        \draw (8.5,-0.866) node [draw,fill,shape=circle,color=black] (u10') {};
        \draw (9.5,-0.866) node [draw,fill,shape=circle,color=black] (u11') {};
        \draw (10,0) node [draw,fill,shape=circle,color=black] (u12) {};      
        \draw (11,0) node [draw,fill,shape=circle,color=black, label=0:{\Large $v$}] (v) {};

        \draw (u9) edge [semithick] (u10)
        (u10) edge [thick, red] (u11)
        (u11) edge [semithick] (u12)
        (u9) edge [semithick] (u10')
        (u10') edge [thick, red] (u11')
        (u11') edge [semithick] (u12)
        (u12) edge [thick, red] (v);

        \draw (5,-1.5) node {Type M};        
      \end{scope}
      \begin{scope}[shift = {(0,-8)}]
        \draw (0,0) node [draw,fill,shape=circle,color=black, label=180:{\Large $u$}] (u) {}; 
        \draw (0.5,0.866) node [draw,fill,shape=circle,color=black] (u1) {};
        \draw (1.5,0.866) node [draw,fill,shape=circle,color=black] (u2) {};
        \draw (0.5,-0.866) node [draw,fill,shape=circle,color=black] (u1') {};
        \draw (1.5,-0.866) node [draw,fill,shape=circle,color=black] (u2') {};
        \draw (2,0) node [draw,fill,shape=circle,color=black] (u3) {};
  
        \draw (u) edge [semithick] (u1)
        (u1) edge [thick, red] (u2)
        (u2) edge [semithick] (u3)
        (u) edge [semithick] (u1')
        (u1') edge [thick, red] (u2')
        (u2') edge [semithick] (u3);
  
        \draw (2.5,0.866) node [draw,fill,shape=circle,color=black] (u4) {};
        \draw (3.5,0.866) node [draw,fill,shape=circle,color=black] (u5) {};
        \draw (2.5,-0.866) node [draw,fill,shape=circle,color=black] (u4') {};
        \draw (3.5,-0.866) node [draw,fill,shape=circle,color=black] (u5') {};
        \draw (4,0) node [draw,fill,shape=circle,color=black] (u6) {};      
  
        \draw (u3) edge [semithick] (u4)
        (u4) edge [thick, red] (u5)
        (u5) edge [semithick] (u6)
        (u3) edge [thick, red] (u4')
        (u4') edge [semithick] (u5')
        (u5') edge [thick, red] (u6);
  
        \draw (5,0) node {\LARGE \dots \dots};
  
        \draw (6,0) node [draw,fill,shape=circle,color=black] (u6') {};      
        \draw (6.5,0.866) node [draw,fill,shape=circle,color=black] (u7) {};
        \draw (7.5,0.866) node [draw,fill,shape=circle,color=black] (u8) {};
        \draw (6.5,-0.866) node [draw,fill,shape=circle,color=black] (u7') {};
        \draw (7.5,-0.866) node [draw,fill,shape=circle,color=black] (u8') {};
        \draw (8,0) node [draw,fill,shape=circle,color=black] (u9) {};      
  
        \draw (u6') edge [semithick] (u7)
        (u7) edge [thick, red] (u8)
        (u8) edge [semithick] (u9)
        (u6') edge [semithick] (u7')
        (u7') edge [thick, red] (u8')
        (u8') edge [semithick] (u9);
  
        \draw (8.5,0.866) node [draw,fill,shape=circle,color=black] (u10) {};
        \draw (9.5,0.866) node [draw,fill,shape=circle,color=black] (u11) {};
        \draw (8.5,-0.866) node [draw,fill,shape=circle,color=black] (u10') {};
        \draw (9.5,-0.866) node [draw,fill,shape=circle,color=black] (u11') {};
        \draw (10,0) node [draw,fill,shape=circle,color=black] (u12) {};      
        \draw (11,0) node [draw,fill,shape=circle,color=black, label=0:{\Large $v$}] (v) {};

        \draw (u9) edge [semithick] (u10)
        (u10) edge [thick, red] (u11)
        (u11) edge [semithick] (u12)
        (u9) edge [thick, red] (u10')
        (u10') edge [semithick] (u11')
        (u11') edge [thick, red] (u12)
        (u12) edge [semithick] (v);

        \draw (5,-1.5) node {Type U};
      \end{scope}
    \end{tikzpicture}  
    \caption{Three types of covers (coloured in red) for the gadget.}
    \label{fig:types-PMU}
  \end{figure}
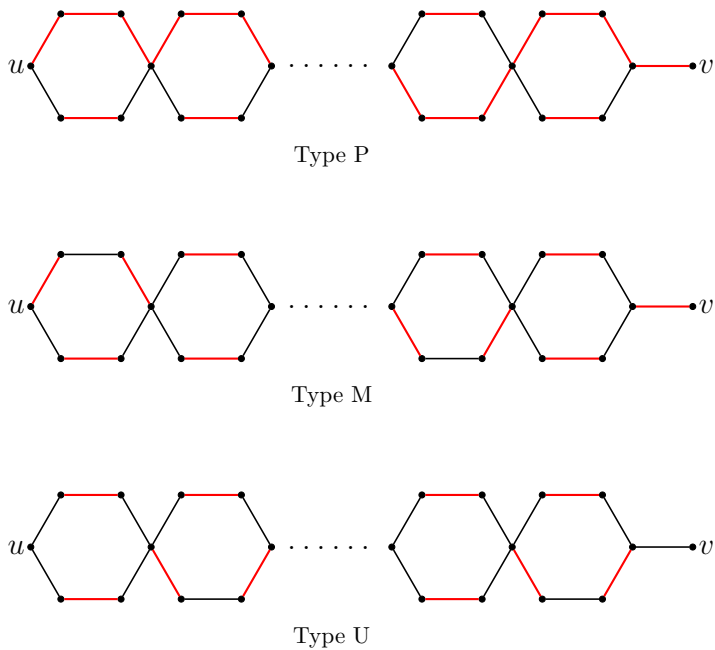  

  If there are $2\ell$ hexagons in the gadget then there are $4^\ell$ configurations of Type~P, $2^\ell$ of type~M
  and $2^\ell$ of type U.  
  
  Let's return to the P2M cover of $G'$.
  The edges in~$G$ corresponding to Type~P configurations in~$G'$ form a collection of disjoint odd cycles in~$G$.  
  (Note that the path in a Type P configuration is of odd length.)  
  Moreover the edges of~$G$ corresponding to Type~$M$ configurations form a collection of isolated edges, which are disjoint from the odd cycles.  
  Finally, the collection of all odd cycles and isolated edges described above cover all the vertices of~$G$, and hence form a P2M cover of~$G$.
  Conversely, any P2M cover in $G$ can be lifted to a P2M cover of $G'$ by choosing a Type~P configuration in~$G'$ for every cycle edge of~$G$, 
  a Type~M configuration in~$G'$ for every isolated edge in $G$, and a Type~U configuration for every other edge of $G$.  
  Suppose the P2M cover of $G$ has $k$ cycle edges, and hence $m-k$ other edges, where $m=\abs{E(G)}$.
  Then the number of P2M covers in $G'$ that correspond to this particular P2M cover in~$G$ is $4^{\ell k}2^{\ell(m-k)}=2^{\ell m}2^{\ell k}$. 
  Thus, for sufficiently large $\ell$, we expect $\frac{\Zptm(G')}{2^{\ell m}2^{\ell n}}$ to be a good approximation to $\Zoc(G)$, where $n=|V(G)|$.

  To be more specific, we will choose $\ell\ge m+2$.
  The discussion above implies that $\Zptmk{n}(G') = 2^{\ell m}2^{\ell n}\Zoc(G)$,
  where $\Zptmk{k}(G')$ denotes the number of P2M covers of~$G'$ with $k$~edges of the Type~P{}.
  Note that $\Zptm(G')=\sum_{k=0}^{n}\Zptmk{k}(G')$.
  Moreover, as $\Zptm(G)\le 2^m$,
  \begin{align}\label{eqn:sum-p2m-n-1}
    \sum_{k=0}^{n-1}\Zptmk{k}(G') \le 2^m 2^{\ell m} 2^{\ell (n-1)} \le \frac{2^{\ell(m+n)}}{4}.
  \end{align}
  Suppose we want to decide whether $\Zoc(G)$ is zero or non-zero, and that we have available an approximation 
$\widetilde{Z}$ to $\Zptm(G')$ that satisfies $\frac12\le \frac{\widetilde{Z}}{\Zptm(G')}\le \frac32$.
  There are two cases.
  \begin{itemize}
    \item If $\widetilde{Z}\ge 2^{\ell(m+n)-1}$,
      then $\Zptm(G')\ge \frac{2^{\ell(m+n)}}{3} > \frac{2^{\ell(m+n)}}{4}$,
      which by \eqref{eqn:sum-p2m-n-1} implies that $\Zptmk{n}(G')> 0$ and hence $\Zoc(G)\ge 1$.
    \item Otherwise $\widetilde{Z} < 2^{\ell(m+n)-1}$,
      and thus $\Zptm(G') < 2\times 2^{\ell(m+n)-1}= 2^{\ell(m+n)}$.
      Since $\Zptmk{n}(G)$ is divisible by $2^{\ell(m+n)}$ it follows that $\Zptmk{n}(G')=0$ and $\Zoc(G)=0$.
  \end{itemize}
Thus we have the desired reduction from an $\NP$-complete decision problem (Hamiltonian path) to the problem of approximating the number of P2M covers in a graph within relative error~2.  It follows immediately that the existence of a FPTAS for $\Zptm(G)$ would imply $\Ptime=\NP$.  Under the weaker assumption that an FPRAS exists, we could,
with probability at least $\frac34$, correctly decide an $\NP$-complete problem.  Thus, any problem in~$\NP$ could be solved in polynomial time by a randomised algorithm with two-sided error or, formally, $\NP\subseteq\BPP$.  It is a standard fact that this last inclusion implies $\RP=\NP$.
\end{proof}

The `powering' construction (\Cref{fig:hex-gadget}) used here is similar to those used in a similar context by Khachiyan~\cite{Khachiyan} and Najt~\cite{Najt}, and more generally by Jerrum, Valiant and Vazirani~\cite{JVV}.  


\section{Transposes of network matrices}\label{sec:transposenetwork}
As we noted, one way to specify a 0/1 polytope is by giving a system of linear inequalities $Ax\leq b$ where the matrix~$A$ is totally unimodular.   Network matrices and their transposes are interesting subclasses of totally unimodular matrices.  In some sense, network matrices and their transposes, together with a certain $5\times 5$ matrix, generate all totally unimodular matrices~\cite{Truemper}.  We begin by defining the class of network matrices. 

Suppose we have a directed graph $(V,E)$ and a directed tree $(V,T)$ on the same vertex set.  (The orientations of the arcs of the tree are arbitrary, and do not necessarily point towards a specific root vertex.  It is convenient to allow parallel arcs in the graph $(V,E)$.)  Given an arc $e=uv\in E$ and an arc $t\in T$, let $\Pi$ be the path from $u$ to $v$ in $T$.  Then define the $|T|\times|E|$ matrix~$A$ by 
$$
A_{te}=\begin{cases}
+1,&\text{if $t$ occurs in a forward direction in $\Pi$;}\\
-1,&\text{if $t$ occurs in a backward direction in $\Pi$;}\\
0,&\text{otherwise}.
\end{cases}
$$
The matrix $A$ is the \emph{network matrix generated by} $(V,E)$ and  $(V,T)$.  An integer matrix is said to be a network matrix if it is generated in this way. 
\begin{proposition}\label{prop:networkimpliesTU}
Any network matrix is totally unimodular.
\end{proposition}
See \cite[Thm 13.20]{SchrijverA} for a proof.  Of course the transpose of a network matrix is also totally unimodular.  The following easy lemma will be useful.  Amongst other things, it tells us that if we have a system of inequalities defined by network matrix, or the transpose of a network matrix, then we can freely add additional bounds such as $\bfz\leq x\leq\bfo$.

\begin{lemma}\label{lem:preserve}
The property of being a network matrix is preserved under the following operations:
\begin{enumerate}[label=(\alph*)]
\item duplicating a row or a column;
\item negating a row or a column;
\item extending the matrix with a unit row or column (one which has a 1 in a single location and zeros elsewhere).
\end{enumerate} 
\end{lemma}

\begin{proof}
Let $A$ be a network matrix defined by the directed graph $(V,E)$ and directed tree $(V,T)$.  Recall that rows of $A$ correspond to tree arcs $t\in T$ and columns to graph arcs $e\in E$.
\begin{enumerate}[label=(\alph*)]
\item To duplicate a row indexed by arc $t=uv\in T$, introduce a new vertex $w$ and set $T:=T\cup\{uw,wv\}\setminus\{uv\}$.  To duplicate a column indexed by arc $e\in E$, introduce a new arc $e'$ parallel to~$e$. 
\item To negate a row indexed by $t$, reverse the direction of arc $t\in T$.  To negate a column indexed by~$e$, reverse the direction of arc $e\in E$.
\item To introduce a new row with a 1 in the column indexed by $e=uv\in E$, introduce a new vertex~$w$ and set $E:=E\cup\{uw\}\setminus\{uv\}$ and $T:= T\cup\{vw\}$.  To introduce a new column with a 1 in the row indexed by $t=uv\in T$, introduce a new arc $uv$ to~$E$.
\end{enumerate}
It is easy to check that these actions have the desired effect on the matrix~$A$.
\end{proof}

So now consider a 0/1~polytope~$P$ defined by inequalities $Ax\leq b$ where $A$ is the \emph{transpose} of a network matrix.  
Let us consider the defining equations of the facets of~$P$ in terms of $(V,E)$ and $(V,T)$.
The arcs in $T$ correspond to variables and those in $E$ to (left-hand sides of) inequalities.  Let $uv\in E$ be an arc in the graph.  The variables that occur in the corresponding inequality are the ones encountered when tracing out the unique path from $u$ to $v$ in $(V,T)$.  
The coefficient of a variable is $+1$ or $-1$ depending on whether the arc in the tree is aligned with or against the direction of the path.  

\begin{example}[independent sets in a bipartite graph]\label{ex:BIS}
  Let $B=(U\cupdot U',F)$ be a bipartite (undirected) graph,
  where $U=\{u_1,\ldots,u_n\}$ and $U'=\{u_1',\ldots,u_m'\}$ are the parts of the bipartition of the vertex set.  
  We encode the independent sets in $B$ as vertices of a polytope defined by the transpose of a network matrix.  As above, we specify this matrix by giving the graph $(V,E)$ and tree $(V,T)$.
  Introduce a new vertex $r$, and let $V=U\cup\{r\}\cup U'$ and $T=\{u_ir:1\leq i\leq n\}\cup\{ru_j':1\leq j\leq m\}$.  
  Let the arc set~$E$ be obtained from~$F$ by simply orienting all edges in~$F$ from $U$ to~$U'$. 
  Each arc $u_ir\in T$ (respectively, $ru_j'\in T$) corresponds to a variable $x_i$ (respectively, $y_j$).  
  Consider the network matrix~$A$ defined by $(V,E)$ and $(V,T)$.
  Each inequality in 
  $\transpose{A}\transpose{(x_1,\ldots,x_n,y_1,\ldots,y_m)}\le\*1$ 
  is of the form $x_i+y_j\leq 1$ for some $u_iu_j'\in F$.  
  Then we add the inequalities $0\leq x_i,y_j\leq 1$ for every $i\in[n]$ and $j\in[m]$,
  which can be done while maintaining the defining matrix to be the transpose of a network matrix,
  by Lemma~\ref{lem:preserve}(c).
  These inequalities define the independent (or stable) set polytope of the bipartite graph~$B$ \cite[Thm.~19.7]{SchrijverA}.  This fact also follows easily from total modularity of the system of inequalities, which itself follows from Proposition~\ref{prop:networkimpliesTU}.
\end{example}

Our goal in this section is to precisely locate the complexity of \nPolytopeVertices, when the matrix $A$ is the transpose of a network matrix.  For the upper bound (which is the non-trivial direction) we use an approximation-preserving reduction to the following problem. 

\begin{description}[itemsep=0pt]
\item [Problem.] \npnSAT.
\item [Instance.]  A CNF Boolean formula $\varphi$ in which each clause contains at most one positive literal and at most one negative literal. 
\item [Output.]  The number of satisfying assignments of $\varphi$.
\end{description}
We know that $\npnSAT$ is equivalent to $\nBIS$ under polynomial-time approximation-preserving reductions~\cite[Thm~5]{Relative}.

In the following theorem we may assume, by Lemma~\ref{lem:preserve} that the matrix~$A$ incorporates the constraints $\bfz\leq x\leq\bfo$.  In light of Proposition~\ref{prop:networkimpliesTU}, this provides the promise demanded by \nPolytopeVertices. 

\begin{theorem}\label{thm:BISequiv}
When $A$ is restricted to be the transpose of a network matrix, the problem \nPolytopeVertices{} is equivalent under polynomial-time approximation-preserving reductions to \nBIS.
\end{theorem}

\begin{proof}
  We have just seen in \Cref{ex:BIS} that \nBIS{} is essentially a special case of \nPolytopeVertices, so we just need to describe a polynomial-time approximation-preserving reduction from \nPolytopeVertices\ to $\npnSAT$ in the case that $A$ is the transpose of a network matrix.  The reduction exploits a construction from Chen, Dyer, Goldberg, Jerrum, Lu, McQuillan and Richerby~\cite[Lem.\ 46]{ChenEtAl}. 

Let the polytope $P$ be an instance of \nPolytopeVertices\ defined by a matrix~$A$ and vector~$b$, where $A$ is the transpose of a network matrix.  Suppose in turn that $A$ is specified by the directed tree $(V,T)$ and directed graph $(V,E)$ on the common vertex set~$V$.   
Recall that variables are associated with arcs in $T$, thus:  $\{x_t:t\in T\}$.  Choose an arbitrary root $r\in V$ as reference point.  The first step is to make a change of variables.  Introduce a new set of variables $\{z_v:v\in V\}$, and define $z_r=0$ and $z_v-z_u=x_t$ for all $t=uv\in T$.  Thus 
\begin{equation}\label{eq:zvdefn}
z_v=\epsilon_{v_0v_1} x_{v_0v_1} + \epsilon_{v_1v_2} x_{v_1v_2}+ \cdots + \epsilon_{v_{\ell-1}v_\ell} x_{v_{\ell-1}v_\ell},
\end{equation}
where $(r=v_0,v_1,\ldots,v_\ell=v)$ is the path from $r$ to $v$ in the tree, and the coefficient $\epsilon_{v_{i-1}v_i}$ associated with the $i$th term is $+1$ if arc $v_{i-1}v_i$ is traversed in the forward direction and $-1$ otherwise.  Note that the variables $\{x_t:t\in T\}$ determine the variables $\{z_v:v\in V\}$ and vice versa.

We now observe that the inequality $L_e(x)\leq b_e$ defined by an arc $e=uv\in E$ takes the simple form $z_v-z_u\leq b_e$, when transformed to the new variables.
By definition, 
\begin{equation}\label{eq:zudefn}
z_u=\epsilon_{u_0u_1}x_{u_0u_1} + \epsilon_{u_1u_2}x_{u_1u_2}+ \cdots + \epsilon_{u_{k-1}u_k}x_{u_{k-1}u_k},
\end{equation}
where $(r=u_0,u_1,\ldots u_k=u)$ is the path from $r$ to $u$ in~$T$, and the signs are determined by the directions of the arcs.  Suppose $u_h=v_h$ is the lowest common ancestor of $u$ and~$v$ in~$T$. Subtracting \eqref{eq:zudefn} from \eqref{eq:zvdefn},
\begin{align*}
z_v-z_u&=\big(\epsilon_{v_0v_1} x_{v_0v_1} + \cdots + \epsilon_{v_{\ell-1}v_\ell} x_{v_{\ell-1}v_\ell}\big)-\big(\epsilon_{u_0u_1}x_{u_0u_1} + \cdots + \epsilon_{u_{k-1}u_k}x_{u_{k-1}u_k}\big)\\
&=\big(\epsilon_{v_hv_{h+1}} x_{v_hv_{h+1}} + \cdots + \epsilon_{v_{\ell-1}v_\ell} x_{v_{\ell-1}v_\ell}\big)-\big(\epsilon_{u_hu_{h+1}}x_{u_hu_{h+1}} + \cdots + \epsilon_{u_{k-1}u_k}x_{u_{k-1}u_k}\big)\\
&=\epsilon_{u_ku_{k-1}}x_{u_ku_{k-1}} + \cdots + \epsilon_{u_{h+1}u_h}x_{u_{h+1}u_h}+\epsilon_{v_hv_{h+1}}x_{v_hv_{h+1}} + \cdots + \epsilon_{v_{\ell-1}v_\ell} x_{v_{\ell-1}v_\ell},
\end{align*}
where the second equality results from cancellation along the common segment of the paths $(u_0,\ldots,u_h)=(v_0,\ldots,v_h)$, and the third from reversing the arcs along the subpath $(u_h,\ldots,u_k)$. (We view the variable $x_{uv}$ as being associated with the undirected edge $uv$, so $x_{uv}$ and $x_{vu}$ denote the same variable.)  But the final line above is just the left-hand side $L_e(x)$ of the linear inequality defined by the arc $e=uv\in E$.  Thus, in the new variables, the inequality reads $z_v-z_u\leq b_e$.  Summarising, integer solutions to the system of equations $\{z_v-z_u\leq b_e:uv=e\in E\}$ are in bijection with integer (and hence 0/1) solutions to $\{L_e(x)\leq b_e:e\in E\}$, which in turn are in bijection with the vertices of polytope~$P$.  (Note that 0/1 points lying in $P$ are necessarily vertices of~$P$.)

So it just remains to encode the inequalities 
\begin{equation}\label{eq:zineqs}
z_r=0\quad\text{and}\quad z_v-z_u\leq b_e,\>\text{for $e=uv\in E$}
\end{equation}
as clauses within an instance~$\varphi$ of $\npnSAT$.

To do this, introduce Boolean variables $\{\zeta_v^i:v\in V\text{ and }-n<i\leq n\}$ and start to build an instance $\varphi$ on this variable set by introducing clauses 
$$
\{\zeta_v^{i+1}\implies \zeta_v^i:v\in V\text{ and }-n<i<n\}.
$$  
Note that the clause $\zeta_v^{i+1}\implies \zeta_v^i$ is logically equivalent to $\neg \zeta_v^{i+1}\vee \zeta_v^i$, so has the correct syntactic form.  Also note that for each $v\in V$ there are $2n+1$ consistent assignments to the variables $\{\zeta_v^i:-n<i\leq n\}$, namely 
$$
(\zeta_v^{-n+1},\zeta_v^{-n+2},\ldots,\zeta_v^{n-1},\zeta_v^n)=\begin{cases}
(0,0,0,\ldots,0,0),\\
(1,0,0,\ldots,0,0),\\
(1,1,0,\ldots,0,0),\\
\qquad\vdots\\
(1,1,1,\ldots,1,0),\\
(1,1,1,\ldots,1,1),
\end{cases}
$$
where we associate false with~0 and true with~1.
We use these Boolean assignments to encode integer assignments in the range $\{-n,\ldots,+n\}$ to~$z_v$ via the correspondence
$$
\zeta_v^i=1\iff i\leq z_v,\quad\text{for all $-n<i\leq n$}.
$$
Note that, by construction, for all $v\in V$,
$$
|z_v|\leq \text{(length of the path in $T$ from $r$ to $v$)}\leq n,
$$
so our encoding covers the feasible range of $z_v$.

Next, we encode the $z$-inequalities \eqref{eq:zineqs} by adding extra clauses to~$\varphi$. Note that every $z$-inequality is of the form $z_v-z_u\leq c$, which is equivalent to the collection of clauses $\{\zeta_v^{i+c}\implies\zeta_u^i:-n<i,i+c\leq n\}$.  Finally, the equality $z_r=0$ is equivalent to the conjunction of clauses 
$$
\zeta_r^{-n+1},\;\zeta_r^{-n+2},\;\ldots,\;\zeta_r^{-1},\;\zeta_r^0,\;\neg\zeta_r^1,\;\neg\zeta_r^2\;\ldots,\;\neg\zeta_r^{n-1},\;\neg\zeta_r^n.
$$
This completes the construction of the instance $\varphi$ of $\npnSAT$.  We see that the number of satisfying assignments to~$\varphi$ is equal to the number of feasible solutions to the $z$-inequalities, which in turn is equal to the number of vertices of the polytope~$P$. The reduction is parsimonious (i.e., preserves the number of solutions) and hence is certainly approximation preserving.  
\end{proof}

We finish the section by noting that counting problems associated with 0/1 polytopes are self-reducible, in the sense that the set of vertices of a 0/1 polytope~$P$ can be expressed as the union of the vertices of two lower-dimensional polytopes $P_0$ and~$P_1$ obtained by intersection with planes of the form $x_n=0$ and $x_n=1$.  For all the counting problems considered here, the polytopes $P_0$ and $P_1$ come from the same class of polytopes as~$P$.  This is an easy observation for general 0/1 polytopes, and follows from Lemma~\ref{lem:preserve} for polytopes defined by (transposes of) network matrices.  For self-reducible problems, approximate counting and (almost) uniform sampling are related by polynomial-time reductions, as observed by Jerrum, Valiant and Vazirani~\cite{JVV}.

\section{Network matrices}
Now consider the defining equations $Ax\leq b$ of a polytope~$P$ when $A$ is a network matrix.  Relative to the previous section, the roles of variables and equations are reversed.  Variables now correspond to arcs in~$E$ and equations to arcs in~$T$.  Fix a tree arc $t\in T$ and define
\begin{align*}
F_t^+&=\big\{e=uv\in E:\text{the path from $u$ to $v$ passes through $t$ is the forward direction}\big\},\\
F_t^-&=\big\{e=uv\in E:\text{the path from $u$ to $v$ passes through $t$ is the backward direction}\big\},
\end{align*} 
where the paths in question are paths in the tree $(V,T)$.
Then the inequality defined by the arc $t$ is 
$$
L_t(x)=\sum_{e\in F_t^+}x_e-\sum_{e\in F_t^-}x_e\leq b_t.
$$
We are interested in the polytope defined by the system $\{L_t\leq b_t:t\in T\}$.

\begin{example}[The matching polytope for a bipartite graph]
  Given a bipartite graph~$B=(U\cupdot U',F)$ where $\abs{U}=\abs{U'}=n$, 
  we set up the same directed graph $(V,E)$ and tree $(V,T)$ as in Example~\ref{ex:BIS} for $\nBIS$.  
  For each edge $u_iu_j'\in F$ of~$B$ we introduce a variable $x_{ij}$ and associate it with the arc $u_iu_j'\in E$.

  Now, the inequality associated with arc $u_ir\in T$ (respectively, $ru_j'\in T$) of the directed tree has the form $\sum_{j:ij\in F}x_{ij}\leq c$ (respectively, $\sum_{i:ij\in F}x_{ij}\leq c$).  
  The defining inequalities of the matching polytope of $B$ are obtained by setting $c=1$ for all edges $u_iu_j\in F$, and adding the inequalities $x_{ij}\geq0$ for all $i,j\in[n]$.  To obtain the perfect matchings polytope, we simply include an inequality $\sum_{j:ij\in F}x_{ij}\geq1$ (respectively $\sum_{i:ij\in F}x_{ij}\geq1$) complementary to each inequality $\sum_{j:ij\in F}x_{ij}\leq1$ (respectively $\sum_{i:ij\in F}x_{ij}\leq1$).  By Lemma~\ref{lem:preserve}, the matrix~$A$ defining this augmented set of inequalities is still a network matrix.  The polytope defined by these inequalities is the (perfect) \emph{matching polytope} of the graph~$B$ \cite[Cor.~18.1b]{SchrijverA};  its vertices correspond to (perfect) matchings in~$B$.  This fact also follows easily from total unimodularity of the matrix~$A$.  
\end{example}

Now imagine that, for each arc $e=uv\in E$, we route $f_e$~units of flow from $u$ to~$v$ in the tree~$T$.  The total flow through tree arc $t\in T$ is then 
$$
\sum_{e\in F_t^+}f_e-\sum_{e\in F_t^-}f_e.
$$
So we can think of the vertex counting problem in terms of integer flows in a network defined on $(V,\overleftarrow{E}\cup T)$.  Arcs in $E$ are reversed (denoted $\overleftarrow{E}$) and have a lower bound of~0 and an upper bound of~1.  Arcs in $T$ have upper bounds defined by the right-hand sides of the inequalities;  thus arc $t\in T$ has an upper bound of~$b_t$.  An integer flow in the network is an assignment of integers (flows) to the arcs $\overleftarrow{E}\cup T$ that satisfies conservation of flow at the each of the vertices in~$V$, and that respects the bounds on each edge.

We claim that integer flows in the above network are in bijection with $0/1$ points in the polytope $P=\{x:\bfz\leq Ax\leq b\}$.  So if $P$ is known to be a 0/1 polytope then the integer flows are in bijection with the vertices of~$P$.  Given a 0/1 assignment to the variables $\{x_{uv}:uv\in E\}$ we construct a flow by forcing $x_{uv}$ units of flow through the arc $vu\in\overleftarrow{E}$.  This causes $L_t(x)$ units of fluid to flow through tree arc $t\in T$.  The resulting flow is legal if and only if $x\in P$.  Also, the above construction clearly gives a bijection between 0/1 points $x$ and legal integer flows in the network.     
Unfortunately, we don't know an FPRAS for integer flows at this level of generality.  However, we do have a positive result for a special case.

\begin{proposition}\label{prop:network}
There is an FPRAS for the following problem:  Given a network matrix~$A$, together with a promise that $P=\{x:\bfz\leq Ax\leq\bfo\}$ is a 0/1 polytope, return the number of vertices of~$P$.  
\end{proposition} 

As before, in light of Lemma~\ref{lem:preserve} we may discharge the promise by adding extra constraints $\bfz\leq x\leq\bfo$ to the system $Ax\leq b$.

\begin{proof}[Proof of Proposition~\ref{prop:network}]
Observe that applying the above construction yields a network in which each arc has capacity~1.  In the case of arcs in $\overleftarrow{E}$, this is by construction.  Arcs in $T$ have lower bound 0 arising from the inequality $Ax\geq\bfz$, and upper bound 1 from the inequality $Ax\leq\bfo$.  In directed graphs where all arcs have capacity~1, Jerrum, Sinclair and Vigoda~\cite[Cor.~8.2]{JSV} show that the number of integral flows can be obtained by reduction to counting perfect matchings in a bipartite graph, for which there is an FPRAS.  Note that since $P$ is a 0/1 polytope, all 0/1 points in~$P$ are actually vertices of~$P$.
\end{proof} 

There seems to be no strong reason to doubt that counting flows in networks with more general bounds admits an FPRAS{}.  However, it appears that this extension of the known FPRAS for counting perfect matchings would require new ideas.  Note that using network matrices we can encode problems such as $b$-matchings and $b$-edge covers in bipartite graphs.  A \emph{$b$-matching} (respectively, \emph{$b$-edge cover}) of a graph is an edge subset that covers each vertex at most (respectively, at least) $b$~times.
For these problems, we have efficient approximation algorithms for some~$b$ but not in general:  see Huang, Lu and Zhang~\cite{winding}. 

\begin{problem}
  Is there an FPRAS for \nPolytopeVertices, subject only to the restriction that $A$ is a network matrix?
\end{problem}


\section{Beyond totally unimodular matrices}\label{sec:TDI}

Many integral polytopes arising in combinatorial optimisation arise from totally unimodular matrices.  In fact, unimodular matrices are the only ones with the property that the polyhedron defined by $Ax\leq b$ is integral for all choices of the integral vector~$b$.  However, if we consider $A$ and~$b$ together, it can happen that the pair $(A,b)$ defines an integral polytope even when $A$ is not totally unimodular.  Since we have not so far discovered any family of $0/1$ polytopes whose vertex-counting problem is harder than $\nBIS$, it is tempting to look beyond totally unimodular. 

A known class of integral polytopes arise from `Totally Dual Integral' (TDI) pairs $(A,b)$~\cite[\S5.17]{SchrijverA}.  A fascinating example is provided by the stable matching polytope of a bipartite graph, which is defined by a natural system of polynomially many inequalities~\cite{VandeVate}.  The defining matrix~$A$ is apparently not totally unimodular, but the linear system $(A,b)$ was shown to be TDI by Kir\'{a}ly and Pap~\cite{KiralyPap}.  Intriguingly, Chebolu, Goldberg and Martin~\cite{ChGM} have shown that the problem of counting stable matchings (and hence the vertices of the stable matching polytope) is interreducible with $\nBIS$ via approximation-preserving reductions.  So, again, we do not manage to get beyond $\nBIS$.  This raises the question of whether the  vertex counting problem considered here is $\nBIS$-easy. 

\begin{problem}
Is \nPolytopeVertices{} approximation-preserving reducible to $\nBIS$ in general?
\end{problem}

\section*{Acknowledgement}

HG would like to thank Elle Najt for sharing the manuscript \cite{Najt} and the pointer to \cite{Khachiyan}.




\bibliographystyle{plain}
\bibliography{polytope-vertex-count}

\end{document}